\documentclass[pra, reprint, letter]{revtex4-1}
\usepackage{fullpage}
\usepackage[T1]{fontenc}
\usepackage{amssymb, amsmath, amsthm}
\makeatletter
\def\amsbb{\use@mathgroup \M@U \symAMSb}
\renewcommand\paragraph{%
   \@startsection{paragraph}{4}{0mm}%
      {-\baselineskip}%
      {.1\baselineskip}%
      {\normalfont\normalsize\bfseries}}
\makeatother
\usepackage[mathscr]{euscript}
\usepackage{mathtools}
\usepackage[colorlinks=true]{hyperref}
\usepackage{url}
\usepackage{breakurl}
\usepackage{verbatim}
\usepackage{ifthen}
\usepackage{framed}
\usepackage{multirow}
\usepackage{bbold}
\usepackage{graphicx}
\usepackage{caption}
\usepackage{subcaption}

\newtheorem{lem}{Lemma}[section]

\newtheorem{cor}{Corollary}[section]

\DeclareMathOperator{\rk}{rk}

\DeclareMathOperator{\tr}{tr}

\DeclareMathOperator{\arctanh}{arctanh}

\newcommand{\ave}[1]{\langle #1 \rangle}

\newcommand{\ket}[1]{| \hspace{1pt} #1 \rangle}
\newcommand{\braket}[2]{\langle #1 \hspace{1pt} | \hspace{1pt} #2 \rangle}

\newcommand{\bramatket}[3]{\langle #1 \hspace{1pt} | #2 | \hspace{1pt} #3 \rangle}
\newcommand{\bramatketq}[2]{\bramatket{#1}{#2}{#1}}

\newcommand{\nbox}[2][9]{\hspace{#1pt} \mbox{#2} \hspace{#1pt}}
\newcommand{\norm}[2][]{#1| \! #1| #2 #1| \! #1|}
\newcommand{\abs}[2][]{#1| #2 #1|}
\newcommand{\tran}[0]{^\textnormal{\tiny{T}}}

\newcommand{\statpro}[2]
{
\begin{equation*} #1 \end{equation*}
\ifthenelse{\equal{\displayproofs}{1}}{\begin{proof} #2 \end{proof}}{}
}

\newcommand{\cH}{\mathcal{H}}

\newcommand{\cS}{\mathcal{S}}

\newcommand{\sH}{\mathscr{H}}

\begin{document}
\title{Entropic uncertainty from effective anti-commutators}
\author{J\k{e}drzej Kaniewski}
\email{j.kaniewski@nus.edu.sg}
\author{Marco Tomamichel}
\author{Stephanie Wehner}
\affiliation{Centre for Quantum Technologies, National University of Singapore, 3 Science Drive 2, Singapore 117543}
\date{\today}
\begin{abstract}
We investigate entropic uncertainty relations for two or more binary measurements, for example spin-$\frac{1}{2}$ or polarisation measurements. We argue that the effective anti-commutators of these measurements, i.e.\ the anti-commutators evaluated on the state prior to measuring, are an expedient measure of measurement incompatibility. Based on the knowledge of pairwise effective anti-commutators we derive a class of entropic uncertainty relations in terms of conditional R\'{e}nyi entropies. 
Our uncertainty relations are formulated in terms of effective measures of incompatibility, which can be certified device-independently. Consequently, we discuss potential applications of our findings to device-independent quantum cryptography.
Moreover, to investigate the tightness of our analysis we consider the simplest (and very well-studied) scenario of two measurements on a qubit. We find that our results outperform the celebrated bound due to Maassen and Uffink [Phys.~Rev.~Lett.~\textbf{60}, 1103 (1988)] and provide a new analytical expression for the minimum uncertainty which also outperforms some recent bounds based on majorisation.
\end{abstract}
\maketitle
\subparagraph*{Introduction.}
Uncertainty relations tell us that quantum mechanics is inherently non-deterministic, i.e.~there exist experiments whose outcomes cannot be predicted with arbitrary precision. In the usual scenario we consider two distinct measurements, giving rise to random variables $X$ and $Y$, respectively, and the statement is of the form: ``if the two measurements are \emph{incompatible} then it cannot be the case that both $X$ and $Y$ are close to being deterministic'' and the statement must hold \emph{regardless of the state of the system prior to measuring}. In other words, $X$ or $Y$ (or both) must be at least somewhat unpredictable, that is, random. To make this statement rigorous we need three ingredients: a measure of incompatibility, a measure of uncertainty and a non-trivial relation between the two.

The study of uncertainty relations began when Heisenberg~\cite{heisenberg27} and (more formally) Kennard~\cite{kennard27} noticed that it is impossible to prepare a particle whose position and momentum are sharply peaked: the more localised a particle is, the more variable its momentum becomes and vice versa. More generally, Robertson~\cite{robertson29} showed that uncertainty might arise whenever two observables do not commute. Let $\rho$ be the state of the system prior to the measurement. For an operator $A$, denote the expectation value of that operator by $\ave{A} = \tr (A \rho)$. For operators $A$ and $B$, let $[A, B] = AB - BA$ be the \emph{commutator} of $A$ and $B$ and let $\ave{[A, B]}$ be the \emph{effective commutator}. Robertson's relation reads
\begin{equation*}
\sigma_{A} \sigma_{B} \geq \frac{1}{2} \abs[\big]{\ave{[A, B]}},
\end{equation*}
where $\sigma_{X}$ is the standard deviation of $X$, $\sigma_{X}^{2} = \ave{X^{2}} - \ave{X}^{2}$. Note that this relation applies to both continuous-outcome (e.g.~position or momentum) and discrete-outcome (e.g.~spin or polarisation) measurements.

In 1930 Schr\"{o}dinger~\cite{schrodinger30} proved a 
stronger relation:
\begin{equation*}
\sigma_{A}^{2} \sigma_{B}^{2} \geq \abs[\Big]{\frac{1}{2} \ave{\{A, B\}} - \ave{A} \ave{B}}^{2} + \abs[\Big]{\frac{1}{2} \ave{[A, B]}}^{2},
\end{equation*}
where $\{A, B\} = AB + BA$ is the \emph{anti-commutator} of $A$ and $B$ and $\ave{\{A, B\}}$ is the \emph{effective anti-commutator}.

These early uncertainty relations are interesting from the foundational point of view but they suffer from two problems: a) they are not tight in some important cases (e.g.~for spin-$\frac{1}{2}$ particle with $A = \sigma_{Z}$, $B = \sigma_{X}$ and $\rho = \frac{\mathbb{1}}{2}$
the right-hand side is $0$, despite both outcomes being maximally random, $\sigma_{A} = \sigma_{B} = 1$) and b) their applications are limited because the standard deviation is not always a suitable measure of uncertainty.

To find uncertainty relations with applications in information theory and cryptography, entropies were employed as measures of uncertainty. Usually one considers a scenario in which we have a certain number of measurements and perform one of them uniformly at random. If we store the label of the measurement in $K$ and the measurement outcome in $X$ we obtain a joint probability distribution $P_{XK}$. Entropic uncertainty relations are simply lower bounds on a particular conditional entropy, $H(X | K)$, evaluated on the probability distribution $P_{XK}$.

The first entropic uncertainty relation was proved for position and momentum of an infinite dimensional system in 1975~\cite{beckner75, bialynicki-birula75} and arguably the most celebrated result
came in 1988~\cite{maassen88}. It states that for two projective rank-1 measurements on a $d$-dimensional, described by measurement eigenvectors $\{\ket{x_{j}}\}_{j \in [d]}$ and $\{\ket{y_{j}}\}_{j \in [d]}$, we have
\begin{equation*}
H(X | K) \geq - \frac{1}{2} \log c,
\end{equation*}
where $H(X | K)$ is the conditional Shannon entropy and $c := \max_{j, k} \abs{\braket{x_{j}}{y_{k}}}^{2}$ is the \emph{overlap} of the two measurements (note that this is independent of the state $\rho$ prior to measurement). Entropic uncertainty relations became an active topic of research since entropies give operational meaning to the notion of uncertainty and thus find applications in many information processing and cryptographic tasks (see~\cite{wehner10} for a recent review).

The authors of~\cite{wehner08} considered a set of binary observables that pairwise anti-commute (as operators) and they found that such measurements give rise to strong entropic uncertainty relations. While the case of perfect anti-commutation is well understood nothing is known about the case of partial (or approximate) anti-commutation. Since for most applications we need uncertainty relations which are robust against small perturbations we turn to study observables which only partially anti-commute as quantified by effective anti-commutators.
\subparagraph*{Results and outline.}
In this paper we prove uncertainty relations for an arbitrary set of binary observables (as usual we associate their outcomes with values $\pm 1$). Given the knowledge of their pairwise effective anti-commutators (cf.~\eqref{eq:def-anti-comm}) we derive lower bounds on conditional R\'{e}nyi entropies (cf.~\eqref{eq:main-result-1} and \eqref{eq:main-result-2}) in two steps. In the first step we show that fixing the effective anti-commutators imposes a simple geometric constraint on the expectation values of these observables (note that a probability distribution with two outcomes is fully characterised by its expectation value). In the second step we show that the constraint on expectation values implies a lower bound on entropic uncertainty.

Our relations have two desirable features.
First, our measure of incompatibility is effective (state-dependent) and it can be certified experimentally based on ideas of Mayers and Yao~\cite{mayers98, mayers04}, which leads to \emph{device-independent uncertainty}. (Note that non-effective measures, like the overlap commonly used in entropic uncertainty relations, cannot be certified and so we can only employ these relations when the device is trusted.) Secondly, we can treat any (finite) number of observables. This is because we do not rely on a standard technique based on a reduction to qubits (Jordan's lemma) which only works for two observables, but instead use the full anti-commutation structure of the set of observables.

We compare our results with existing bounds for the case of the Shannon entropy of two measurements. In particular, we improve on the celebrated Maassen-Uffink bound by providing an analytical bound that is strictly stronger for all non-trivial overlaps. We conclude the paper with a discussion of potential applications to device-independent quantum cryptography.
\subparagraph*{Techniques.}
A binary measurement consists of two positive semi-definite operators, $F_{+}, F_{-} \geq 0$, that add up to identity, $F_{+} + F_{-} = \mathbb{1}$. If we associate the outcomes with values $\pm 1$ then the measurement can be written compactly as a binary \emph{observable}, $A = F_{+} - F_{-}$, which satisfies $- \mathbb{1} \leq A \leq \mathbb{1}$.

Suppose we are given a state, $\rho$, and a set of $M$ binary observables, $\{A_{j}\}_{j \in [M]}$. Define the effective anti-commutator between the $j$-th and the $k$-th observable as 
\begin{equation}
  \label{eq:def-anti-comm}
  \varepsilon_{jk} = \frac{\ave{\{A_{j}, A_{k}\}}}2 = \frac{\tr(\{A_{j}, A_{k}\} \rho)}2
\end{equation}
and note that $\varepsilon_{jk}$ is real and $\abs{\varepsilon_{jk}} \leq 1$. Let $T$ be the \emph{anti-commutation matrix}, $[T]_{jk} = \varepsilon_{jk}$. For ease of presentation in the main paper we focus on projective observables, for which $[T]_{jj} = 1$ for all $j$. For a more general proof, which also covers generalised measurements, please refer to Section~\ref{sec:ellipsoid} of the Supplemental Material (SM). Let $g_{j} = \ave{A_{j}}$ be the expectation value of the $j$-th observable. For binary observables the probability distribution of interest (as described in the introduction) can be written as
\begin{equation}
\label{eq:probability-distribution}
\Pr[X = x, K = k] = \frac{1}{M} \cdot \frac{1}{2} \big(1 + (-1)^{x} g_{k} \big).
\end{equation}
The conditional R\'{e}nyi entropy~\cite{arimoto77} of order $\alpha > 1$ is defined as
\begin{equation}
\label{eq:entropy-definition}
H_{\alpha}(X | K) := \frac{\alpha}{1 - \alpha} \log \sum_{k} p_{k} \Big(\sum_{x} p_{x | k}^{\alpha} \Big)^{1/\alpha}
\end{equation}
while the Shannon entropy equals $H(X | K) := \lim_{\alpha \to 1} H_{\alpha}(X | K)$. The goal is to prove lower bounds on $H_{\alpha}(X | K)$ and $H(X | K)$ (this is what we want) evaluated on the joint probability distribution~\eqref{eq:probability-distribution} based on the knowledge of $T$ (this is what we are given) and we do it in two steps. First, we show that $T$ imposes a geometric condition on the expectation values of the observables. Then, we use this geometric condition to prove lower bounds on entropic uncertainty.

Let $g = (g_{1}, g_{2}, \ldots, g_{M})$ be a (column) vector composed of expectation values. Clearly, $g$ lies inside the $(\pm 1)$-hypercube, $g \in [-1,1]^{M}$, but we show that $T$ imposes an extra geometric constraint on $g$. For this purpose, let $a$ be an arbitrary real unit vector, $a \in [-1, 1]^{M}$, and let $K = \sum_{j} a_{j} A_{j}$.
Then
\begin{equation*}
K^{2} = \mathbb{1} + \frac{1}{2} \sum_{j \neq k} a_{j} a_{k} \{A_{j}, A_{k}\}.
\end{equation*}
For arbitrary operators the Cauchy-Schwarz inequality ensures that $[\tr (X^{\dagger} Y)]^{2} \leq \tr (X^{\dagger} X) \cdot \tr (Y^{\dagger} Y)$. By setting $X = K \sqrt{\rho}$ and $Y = \sqrt{\rho}$ we find that
\begin{equation*}
a \tran g g \tran a \leq a \tran T a.
\end{equation*}
Since this inequality holds for all choices of $a$, it is equivalent to the operator inequality
\begin{equation}
\label{eq:ellipsoid}
g g \tran \leq T.
\end{equation}
This constraint admits an appealing geometrical interpretation: the matrix $T$ defines an ellipsoid within the hypercube and the constraint restricts the vector $g$ to lie inside that ellipsoid (see FIG.~\ref{fig:ellipses} for an example).
\begin{figure}[h]
	\hspace{-0.5cm}%
	\includegraphics{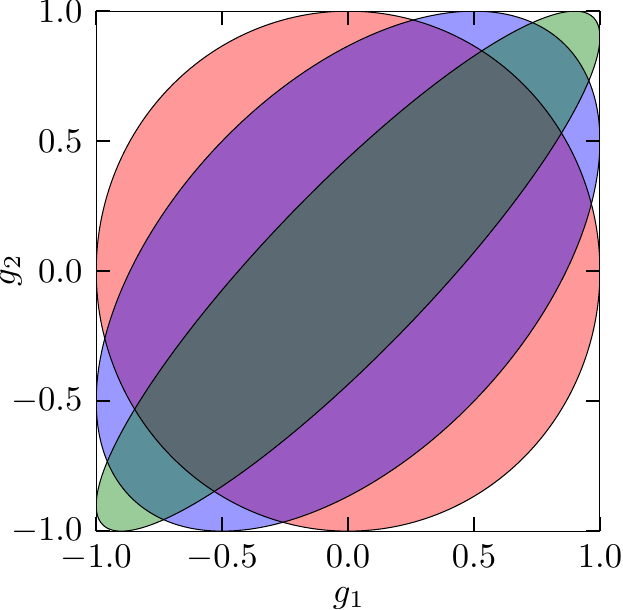}
	\caption{The allowed expectation values of two observables with a fixed effective anti-commutator, $\varepsilon \in \{0, 0.5, 0.9\}$. For $\varepsilon = 0$ we get a circle, which becomes gradually elongated towards the corners as $\varepsilon$ increases. Note that $\varepsilon > 0$ ($\varepsilon < 0$) forces the two expectation values to be correlated (anti-correlated), which results in an ellipse lying along the primary (secondary) diagonal. The deterministic points, corresponding to the corners, are only allowed for $\abs{\varepsilon} = 1$.}
	\label{fig:ellipses}
\end{figure}
Moreover, an extension of the construction from~\cite{tsirelson80} (Section~\ref{sec:ellipsoid} of the SM) shows that this characterisation is tight: a vector of expectation values $g$ and an anti-commutation matrix $T$ are compatible iff~\eqref{eq:ellipsoid} holds.

To find lower bounds on a particular entropy ($H_{\alpha}(X | K)$ or $H(X | K)$) we just need to minimise it over the allowed set of expectation values. Note that for the probability distribution~\eqref{eq:probability-distribution} the expression~\eqref{eq:entropy-definition} simplifies to
\begin{equation*}
H_{\alpha}(X | K) = \frac{\alpha}{1 - \alpha} \log \frac{\sum_{k} w_{\alpha}(g_{k})}{M},
\end{equation*}
where $w_{\alpha}(g) = \big[ \big( \frac{1 + g}{2} \big)^{\alpha} + \big( \frac{1 - g}{2} \big)^{\alpha} \big]^{1/\alpha}$. Now, the task is to minimise $H_{\alpha}(X | K)$ over the ellipsoid, or, equivalently, to solve
\begin{equation*}
\textnormal{max:}\ \sum_{k} w_{\alpha}(g_{k}) \nbox{s.t.} g g \tran \leq T.
\end{equation*}
Unfortunately, this seemingly natural task turns out to be rather difficult even in the simplest cases. Therefore, we consider a relaxation of the problem, in which we optimise over a sphere whose radius is determined by the largest semi-axis of the ellipsoid, denoted by $r$ (see FIG.~\ref{fig:relaxation} for an example). Note that $r = \norm{T}$, the spectral norm of $T$.
\begin{figure}[h]
	\hspace{-0.5cm}%
	\includegraphics{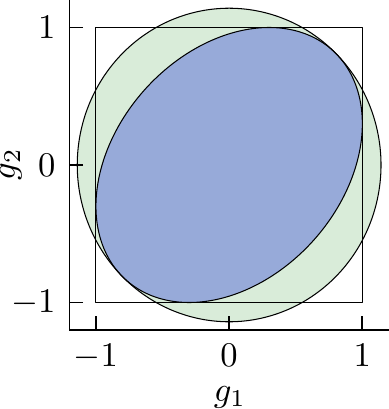}
	\includegraphics{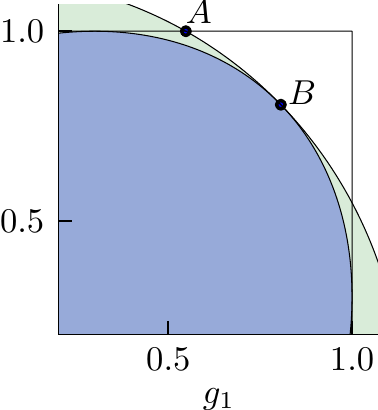}
	\caption{The spherical relaxation for two measurements with $\varepsilon = 0.3$. Optimisation is performed over a circle (light colour) rather than an ellipse (dark colour). Points $A$ and $B$ are the optimal solutions to the relaxed optimisation problem~\eqref{eq:relaxed-problem-2} for convex ($\alpha \in (1, \frac{3}{2}]$) and concave ($\alpha \in [2, \infty)$) functions, respectively.}
	\label{fig:relaxation}
\end{figure}
\begin{equation}
\label{eq:relaxed-problem-1}
\textnormal{max:} \sum_{k} w_{\alpha}(g_{k}) \nbox[6]{s.t.} g \in [-1, 1]^{M}, \ \sum_{k} g_{k}^{2} \leq r.
\end{equation}
Note that we added the hypercube constraints explicitly since it is not implied by the relaxed, spherical constraint. This approach has the advantage that it compresses the whole anti-commutation matrix into just one number\,---\,its norm. More importantly, the relaxed problem can be solved analytically for most values of $\alpha$ as explained below.

Since neither the objective function nor the constraints of~\eqref{eq:relaxed-problem-1} depend on the sign of $g_{k}$ we can restrict ourselves to non-negative expectation values. This allows us to define $t_{k} = g_{k}^{2}$ and the problem becomes
\begin{equation}
\label{eq:relaxed-problem-2}
\textnormal{max:}\ \sum_{y} w_{\alpha}(\sqrt{t_{k}}) \nbox[6]{s.t.} t \in [0, 1]^{M}, \ \sum_{k} t_{k} \leq r.
\end{equation}
Since the objective function is monotone we can assert that the optimal solution satisfies $\sum_{k} t_{k} = r$.

For $\alpha \in (1, \frac{3}{2}]$ the function $w_{\alpha}(\sqrt{t})$ is convex in $t$ (Section~\ref{sec:conv-conc} of the SM) and since the maximum of a convex function over a convex set is achieved at an extremal point, the optimal value must be achieved at an assignment of the form
\begin{equation}
\label{eq:optimal-assignment}
t_{k} =
\begin{cases}
1 &\nbox{for} 1 \leq k \leq \lfloor r \rfloor,\\
r - \lfloor r \rfloor &\nbox{for} k = \lfloor r \rfloor + 1,\\
0 &\nbox{otherwise.}
\end{cases}
\end{equation}
Hence, for $\alpha \in (1, \frac{3}{2}]$ we arrive at the following bound, which constitutes our main result:
\begin{equation}
\label{eq:main-result-1}
\begin{split}
H_{\alpha}(X | K) \geq H_{\alpha}(Y | K), \nbox[8]{where}\\
\Pr[Y = y, K = k] = \frac{1}{M} \cdot \frac{1}{2} \big( 1 + (-1)^{y} \sqrt{t_{k}}\, \big)
\end{split}
\end{equation}
and $t_{k}$ refers to the optimal assignment~\eqref{eq:optimal-assignment}. This can be extended to the Shannon entropy by taking the limit of $\alpha \to 1$ yielding $H(X | K) \geq H(Y | K)$.

For $\alpha \in [2, \infty)$ the function $w_{\alpha}(\sqrt{t})$ is concave and since it is also symmetric the minimum is achieved for $t_{k} = \frac{r}{M}$ for all $k$. Therefore, we have
\begin{equation}
\label{eq:main-result-2}
\begin{split}
H_{\alpha}(X | K) \geq H_{\alpha}(Y), \nbox[8]{where}\\
\Pr[Y = y] = \frac{1}{2} \bigg( 1 + (-1)^{y} \sqrt{\frac{r}{M}}\, \bigg).
\end{split}
\end{equation}
Note that in both cases these bounds are functions of $M$ and $r$ only and, hence, can be computed easily.
\subparagraph*{Comparison with existing bounds.}
Although effective anti-commutators play a central role in our work, it is more common to state uncertainty relations in terms of the overlap. Let us consider two projective rank-1 measurements on a qubit and the conditional Shannon entropy that arises. We look for bounds of the form $H(X | K) \geq q(c)$ and, as stated in the introduction, the celebrated result of Maassen and Uffink~\cite{maassen88} reads
\begin{equation*}
q_{\textnormal{MU}}(c) = - \frac{1}{2} \log c.
\end{equation*}
While this is known to be tight for the extreme values of the overlap, $c \in \{\frac{1}{2}, 1\}$, it is not tight in the interior. It turns out that our results might be applied to this case to give an improvement for all intermediate values of $c$. We take advantage of the fact that for projective measurements on a qubit there is a one-to-one mapping between the effective anti-commutator and the overlap, $c = (1 + \abs{\varepsilon})/2$. Therefore, we can formulate our bound~\eqref{eq:main-result-1} as a function of the overlap
\begin{equation*}
h \bigg( \frac{1 + \sqrt{\abs{\varepsilon}}}{2} \bigg) = h \bigg( \frac{1 + \sqrt{2c - 1}}{2} \bigg) = q_{\textnormal{ac}}(c),
\end{equation*}
where $h(p) = - p \log p - (1 - p) \log (1 - p)$ is the binary entropy.
Moreover, we compare these bounds with a bound recently developed using a majorisation technique~\cite{friedland13, puchala13} (and very recently~\cite{rudnicki14}), denoted $q_{\textnormal{maj}}(c)$, and the largest state-independent lower bound, denoted $q_{\textnormal{opt}}(c)$. (For $c \gtrsim 0.7$ there is an analytic expression for $q_{\textnormal{opt}}$ due to Ghirardi et al.~\cite{ghirardi03}, while for $c \lesssim 0.7$ one needs to resort to numerics.)
\begin{figure}[h]
\includegraphics{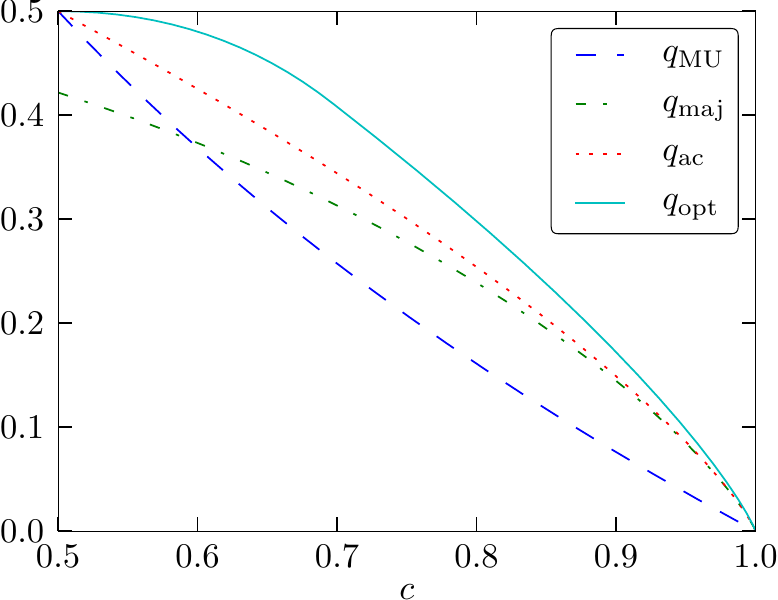}
\caption{Comparison of various lower bounds on $H(X | K)$ as a function of the overlap, $c$.
}
\label{fig:comparison}
\end{figure}
\subparagraph*{Applications to quantum cryptography.}
Recently, in the context of quantum cryptography, there has been a lot of interest in self-testing~\cite{mayers98, mayers04, mckague11} and device-independent security~\cite{acin06}. In self-testing the task is to characterise the internal working of a device 
by analysing observed correlations alone. This characterisation then allows to prove security of a cryptographic protocol executed using that device.
(The term device-independent comes from the fact that we did not assume how the device works but we deduced it from the statistics.)

Uncertainty relations constitute an important ingredient of many device-independent security proofs (see~\cite{lim13} for an example in quantum key distribution and~\cite{miller14} for a very recent example in randomness expansion). An interesting development would be to prove device-independent security for two-party cryptography, for example in the bounded~\cite{damgard05, damgard07} or noisy~\cite{wehner08a, konig12} storage model. (In case of trusted devices, security based on uncertainty relations was proved in the bounded storage model~\cite{ng12} and for relativistic bit commitment~\cite{kaniewski13}.)

Our results fit into this framework since we derive uncertainty from effective anti-commutators, which can be certified experimentally. To certify effective anti-commutation between two observables it is enough to observe Clauser-Horne-Shimony-Holt (CHSH) violation (see, e.g.~\cite{tomamichel13a}). To extend this result to multiple observables we resort to a game proposed by Slofstra~\cite{slofstra11}, which can be seen as a combination of multiple CHSH games in which one of the parties is not told which particular subgame they are playing (see Section~\ref{sec:certification-procedure} of the SM for details). This testing procedure produces bounds on the effective anti-commutator of every pair of observables, which implies an upper bound on the norm of the anti-commutation matrix, $r$. Then, we use \eqref{eq:main-result-1} and \eqref{eq:main-result-2} to obtain explicit entropic bounds, hence, leading us to \emph{device-independent uncertainty}.
\subparagraph*{Conclusion.} Drawing from early uncertainty relations we have shown that it is possible to derive entropic uncertainty relations for binary observables from effective anti-commutation. The effective anti-commutators seem to be a natural object to study and give rise to strong uncertainty relations. Moreover, since they can be certified (self-tested) our uncertainty relations are expected to have applications in device-independent cryptography. Investigating these potential applications is the most interesting open question arising from our research. Another, more foundational line of research could investigate whether our approach can be extended to allow for quantum side information.
\subparagraph*{Acknowledgements.}
This work is funded by the Ministry of Education (MOE) and National Research Foundation Singapore, as well as MOE Tier 3 Grant "Random numbers from quantum processes" (MOE2012-T3-1-009).
We thank Thomas Vidick and Matthew McKague for useful discussions.
\bibliographystyle{apsrev4-1}
\bibliography{librarysan}
\onecolumngrid
\appendix
\section{Preliminaries}
\label{sec:preliminaries}
For an integer $n$, let $[n] = \{1, 2, \ldots, n\}$. Let $\sH$ denote a finite-dimensional Hilbert space of dimension $d = \dim \sH$ and let $\cH(\sH)$ denote the set of Hermitian operators acting on $\sH$. Let $\cS(\sH)$ denote the set of quantum states on $\sH$: $\rho \in \cS(\sH) \iff \rho \in \cH(\sH), \rho \geq 0, \tr \rho = 1$. A binary observable, $\Gamma \in \cH(\sH)$, is a Hermitian operator which satisfies $- \mathbb{1}_{d} \leq \Gamma \leq \mathbb{1}_{d}$, where $\mathbb{1}_{d}$ denotes the identity matrix of dimension $d$.

Let $\{\Gamma_{j}\}$ be a set of Hermitian, traceless, anti-commuting observables acting on a $d$-dimensional Hilbert space:
\begin{equation*}
\Gamma_{j} = \Gamma_{j}^{\dagger}, \quad \tr \Gamma_{j} = 0 \nbox{and} \{\Gamma_{j}, \Gamma_{k}\} = 2 \delta_{jk} \cdot \mathbb{1}_{d}.
\end{equation*}
Note that such a set can always be found, regardless of the number of observables required, as long as the dimension is high enough (e.g.~by Jordan-Wigner transformation, see~\cite{wehner08} for details). Let us first show that if we build a quantum state out of these operators then a simple Bloch-sphere-type condition holds.
\begin{lem}
Let $x$ be a real vector. The operator
\begin{equation*}
\rho = \frac{1}{d} \cdot \Big( \mathbb{1}_{d} + \sum_{j} x_{j} \Gamma_{j} \Big)
\end{equation*}
corresponds to a valid quantum state iff $\sum_{j} x_{j}^{2} \leq 1$.
\end{lem}
\begin{proof}
Clearly, $\rho$ is Hermitian and of unit trace, hence, we just need to verify that it is also positive semi-definite. Let $F = \sum_{j} x_{j} \Gamma_{j}$ and note that $F^{2}$ is proportional to $\mathbb{1}_{d}$. Therefore, $F$ can be written as
\begin{equation*}
F = \big( \sum_{j} x_{j}^{2} \big)^{1/2} \cdot (2 P - \mathbb{1}_{d}),
\end{equation*}
where $P$ is a $d/2$-dimensional projector, $\tr P = d/2$ (note that this implies that $d$ must be even). Clearly, $\rho \geq 0$ is equivalent to $\mathbb{1}_{d} + F \geq 0$, which is satisfied iff $\sum_{j} x_{j}^{2} \leq 1$.
\end{proof}

\section{The ellipsoid condition}
\label{sec:ellipsoid}
Suppose we are given a state, $\rho$, and a set of $M$ binary observables, $\{A_{j}\}_{j \in [M]}$. Let $g$ be the (column) vector of expectation values, $g_{j} = \ave{A_{j}}$, and let $T$ be the anti-commutation matrix
\begin{equation*}
T_{jk} =
\begin{cases}
\ave{A_{j}^{2}} &\nbox{if} j = k,\\
\ave{\{A_{j}, A_{k} \}}/2 &\nbox{otherwise.}
\end{cases}
\end{equation*}
\begin{lem}
Any valid combination of $g$ and $T$ satisfies $g g \tran \leq T$.
\end{lem}
\begin{proof}
Let $a$ be an arbitrary real unit vector, $a \in [-1, 1]^{M}$, and let $K = \sum_{j} a_{j} A_{j}$.
Then
\begin{equation*}
K^{2} = \sum_{j} a_{j}^{2} A_{j}^{2} + \frac{1}{2} \sum_{j \neq k} a_{j} a_{k} \{A_{j}, A_{k}\}.
\end{equation*}
Consider Hermitian operators $X = K \sqrt{\rho}$ and $Y = \sqrt{\rho}$. Note that $\tr(X^{\dagger} Y) = \tr (K \rho) = \sum_{j} a_{j} g_{j} = a \tran g$, $\tr (X^{\dagger} X) = \tr( K^{2} \rho) = a \tran T a$ and $\tr (Y^{\dagger} Y) = \tr (\rho) = 1$. Therefore, the Cauchy-Schwarz inequality, $[\tr (X^{\dagger} Y)]^{2} \leq \tr (X^{\dagger} X) \cdot \tr (Y^{\dagger} Y)$, implies that
\begin{equation*}
a \tran g g \tran a \leq a \tran T a.
\end{equation*}
Since this inequality holds for all choices of $a$, it is equivalent to the operator inequality
\begin{equation*}
g g \tran \leq T.
\end{equation*}
\end{proof}
\begin{lem}
Let $T$ be a $M \times M$ real, positive semi-definite matrix and let $g \in [-1, 1]^{M}$ be a real vector such that $g g \tran \leq T$. Then, there exists a quantum state and measurements that give $g$ as the vector of expectation values and $T$ as the anti-commutation matrix.
\end{lem}
\begin{proof}
Since $T \geq 0$ there exists a $M \times r$ real matrix $R$, such that $R R \tran = T$ and $r = \rk(T)$. Let the $j$-th observable be
\begin{gather*}
A_{j} = \sum_{i = 1}^{r} R_{ji} \Gamma_{i},
\end{gather*}
which implies that $\{A_{j}, A_{k}\} = 2 T_{jk} \cdot \mathbb{1}_{d}$. Therefore, the anti-commutation matrix is reproduced correctly independent of the state.

Consider an operator defined as
\begin{equation*}
\rho = \frac{1}{d} \cdot \Big( \mathbb{1}_{d} + \sum_{j = 1}^{r} x_{j} \Gamma_{j} \Big).
\end{equation*}
It is easy to verify that if $\rho$ corresponds to a valid state then the resulting vector of expectation values equals $g = R x$. Since $\rk(R) = r$, $R$ has a left inverse, namely a $r \times M$ matrix $Q$ such that $Q R = \mathbb{1}_{r}$, and $x$ can be calculated as $x = Q g$. To verify that $\rho$ corresponds to a valid state we must check that $x \tran x \leq 1$ which follows directly from the fact that
\begin{equation*}
x x \tran = Q g g \tran Q \tran \leq Q T Q \tran = Q R R \tran Q \tran = \mathbb{1}_{r},
\end{equation*}
where we used the assumption $g g \tran \leq T$.
\end{proof}
As a corollary we obtain a lower bound on the dimension of the system necessary to reproduce a particular choice of $g$ and $T$.
\begin{cor}
To reproduce correctly $g$ and $T$ it is sufficient to use $r = \rk(T)$ anti-commuting observables which can be realised in dimension $d = 2^{\lceil \frac{r - 1}{2} \rceil}$.
\end{cor}
\section{Convexity/concavity of $w_{\alpha}(\sqrt{t})$}
\label{sec:conv-conc}
For completeness recall the definition of $w_{\alpha}(x)$ for $x \in [-1, 1]$:
\begin{equation*}
w_{\alpha}(x) = \Big[ \Big( \frac{1 + x}{2} \Big)^{\alpha} + \Big( \frac{1 - x}{2} \Big)^{\alpha} \Big]^{1/\alpha}.
\end{equation*}
\begin{lem}
The function $w_{\alpha}(\sqrt{t})$ for $t \in [0, 1]$ is convex for $\alpha \in (1, \frac{3}{2}]$ and concave for $\alpha \in [2, \infty)$.
\end{lem}
\begin{proof}
Let us write $w_{\alpha}(\sqrt{t})$ as
\begin{gather*}
w_{\alpha}(\sqrt{t}) = \frac{1}{2} \big[ g_{\alpha}(t) \big]^{1/\alpha},\\
\nbox{where} g_{\alpha}(t) = (1 + \sqrt{t})^{\alpha} + (1 - \sqrt{t})^{\alpha}.
\end{gather*}
Calculating the derivatives gives
\begin{gather*}
\frac{d}{dt} w_{\alpha}(\sqrt{t}) = \frac{1}{2 \alpha} \cdot g_{\alpha}(t)^{(1 - \alpha)/\alpha} \cdot g' _{\alpha}(t),\\
\frac{d^{2}}{dt^{2}} w_{\alpha}(\sqrt{t}) = \frac{1 - \alpha}{2\alpha^{2}} \cdot g_{\alpha}(t)^{(1 - 2\alpha)/\alpha} \cdot \big[ g' _{\alpha}(t) \big]^{2} + \frac{1}{2\alpha} \cdot g_{\alpha}(t)^{(1 - \alpha)/\alpha} \cdot g'' _{\alpha}(t)\\
= \frac{g_{\alpha}(t)^{(1 - 2\alpha)/\alpha}}{2\alpha^{2}} \cdot \Big[ (1 - \alpha) \cdot \big[ g' _{\alpha}(t) \big]^{2} + \alpha \cdot g_{\alpha}(t) \cdot g'' _{\alpha}(t) \Big].
\end{gather*}
Therefore, what we are interested in is the sign of
\begin{equation}
\label{eq:h-expression}
h_{\alpha}(t) = \frac{1 - \alpha}{\alpha^{2}} \cdot \big[ g' _{\alpha}(t) \big]^{2} + \frac{1}{\alpha} \cdot g_{\alpha}(t) \cdot g''_{\alpha}(t).
\end{equation}
It is easy to verify that
\begin{gather*}
g'_{\alpha}(t) = \frac{\alpha}{2 \sqrt{t}} \Big[ (1 + \sqrt{t})^{\alpha - 1} - (1 - \sqrt{t})^{\alpha - 1} \Big],\\
g''_{\alpha}(t) = \frac{\alpha (\alpha - 1)}{4 t} g_{\alpha - 2}(t) - \frac{g'_{\alpha}(t)}{2t}.
\end{gather*}
Expanding the terms gives
\begin{align*}
\frac{1 - \alpha}{\alpha^{2}} \cdot \big[ g' _{\alpha}(t) \big]^{2} &= \frac{1 - \alpha}{4t} \cdot \Big[ (1 + \sqrt{t})^{2 (\alpha - 1)} + (1 - \sqrt{t})^{2 (\alpha - 1)} - 2 (1 - t)^{\alpha - 1} \Big],\\
\frac{1}{\alpha} \cdot g_{\alpha}(t) \cdot g''_{\alpha}(t) &= \frac{\alpha - 1}{4t} \cdot g_{\alpha}(t) g_{\alpha - 2}(t) - \frac{1}{2 \alpha t} g_{\alpha}(t) g'_{\alpha}(t)\\
&= \frac{\alpha - 1}{4t} \cdot \Big[ (1 + \sqrt{t})^{2 (\alpha - 1)} + (1 - \sqrt{t})^{2 (\alpha - 1)} + 2 (1 + t) (1 - t)^{\alpha - 2} \Big]\\
&- \frac{1}{4 t \sqrt{t}} \cdot \Big[ (1 + \sqrt{t})^{2 \alpha - 1} - (1 - \sqrt{t})^{2 \alpha - 1} - 2 \sqrt{t} (1 - t)^{\alpha - 1} \Big].
\end{align*}
Therefore,
\begin{equation*}
h_{\alpha}(t) = \frac{1}{t} \cdot \bigg( (1 - t)^{\alpha - 2} \Big[ \alpha - \frac{1 + t}{2} \Big] - \frac{1}{4 \sqrt{t}} \Big[ (1 + \sqrt{t})^{2 \alpha - 1} - (1 - \sqrt{t})^{2 \alpha - 1} \Big] \bigg).
\end{equation*}
Since we are only interested in the sign of~\eqref{eq:h-expression}, we consider
\begin{equation}
\label{eq:simplified-expression}
2 \alpha - 1 - t - \frac{(1 - t)^{3/2}}{2 \sqrt{t}} \bigg[ \Big( \frac{1 + \sqrt{t}}{1 - \sqrt{t}} \Big)^{\alpha - 1/2} - \Big( \frac{1 - \sqrt{t}}{1 + \sqrt{t}} \Big)^{\alpha - 1/2} \bigg].
\end{equation}
Here, it is convenient to introduce hyperbolic functions. Let $e^{2x} = (1 + \sqrt{t})/(1 - \sqrt{t})$, which means that $t \in [0, 1]$ is mapped onto $x \in [0, \infty)$. Then, we have
\begin{equation*}
x = \arctanh \sqrt{t}, \quad t = \tanh^{2} x \nbox{and} 1 - t = \frac{1}{\cosh^{2} x}
\end{equation*}
and~\eqref{eq:simplified-expression} becomes
\begin{align*}
2 \alpha &- 1 - \tanh^{2} x - \frac{\sinh [x(2 \alpha - 1)]}{\sinh x \cdot \cosh^{2} x}\\
&= 2 (\alpha - 1) + \frac{\sinh x - \sinh [x(2 \alpha - 1)]}{\sinh x \cosh^{2} x}.
\end{align*}
Note that $2 \sinh x \cosh^{2} x = \sinh 2x \cosh x = (\sinh 3x + \sinh x)/2$. The sign is the same as the sign of
\begin{equation*}
\frac{\alpha - 1}{2} \sinh 3x + \frac{1 + \alpha}{2} \sinh x - \sinh [x(2 \alpha - 1)],
\end{equation*}
which we can Taylor-expand. Note that this is an odd function and the coefficients are
\begin{equation*}
c_{k}(\alpha) = \frac{1}{2 k!} \Big[ (\alpha - 1) \cdot 3^{k} + 1 + \alpha - 2 \cdot (2 \alpha - 1)^{k} \Big].
\end{equation*}
To show convexity (concavity) it suffices to show that all the coefficients are positive (negative).
Since $c_{k}(\alpha)$ is a polynomial and it vanishes at $\alpha = 1$ it must be divisible by $(\alpha - 1)$.
\begin{gather*}
(2 \alpha - 1)^{k} = \sum_{j = 0}^{k} {k \choose j} (\alpha - 1)^{j} \alpha^{k - j} = \alpha^{k} + (\alpha - 1) \sum_{j = 0}^{k - 1} {k \choose j + 1} (\alpha - 1)^{j} \alpha^{k - j - 1},\\
1 + \alpha - 2 \alpha^{k} = (1 - \alpha) + 2 \alpha (1 - \alpha^{k - 1}) = (1 - \alpha) \Big(1 + 2 \sum_{j = 1}^{k - 1} \alpha^{j} \Big),
\end{gather*}
Putting everything together gives
\begin{gather*}
c_{k}(\alpha) = \frac{\alpha - 1}{2 k!} \cdot p_{k}(\alpha),\\
\nbox{where} p_{k}(\alpha) = 3^{k} - 1 - 2 \sum_{j = 1}^{k - 1} \alpha^{j} - 2 \sum_{j = 0}^{k - 1} {k \choose j + 1} (\alpha - 1)^{j} \alpha^{k - j - 1}.
\end{gather*}
Note that for $\alpha \geq 1$, $p_{k}(\alpha)$ is monotonically decreasing in $\alpha$, so it has at most one zero. Therefore, $c_{k}(\alpha)$ has at most two zeroes (the first one at $\alpha = 1$). By checking
\begin{gather*}
c_{k} \Big( \frac{3}{2} \Big) = \frac{1}{2 k!} \Big(\frac{3^{k} + 5}{2} - 2^{k + 1} \Big) \geq 0,\\
c_{k} (2) = \frac{1}{2 k!} (3 - 3^{k}) \leq 0,
\end{gather*}
we conclude that the other zero is always there and is contained within $\alpha \in (\frac{3}{2}, 2)$. Hence for $\alpha \in (1, \frac{3}{2}] \cup [2, \infty)$ all the coefficients have the same sign which proves convexity/concavity of the original function.
\end{proof}
\section{The certification procedure}
\label{sec:certification-procedure}
This certification procedure assumes that both devices are memoryless, i.e.~every round is identical and independent of each other.

Suppose we are given a measurement device (Alice) with $M$ different settings, which correspond to different binary observables, $\{A_{j}\}_{j \in [M]}$. The goal of the certification procedure is to characterise the anti-commutation matrix $T$, or more specifically the effective pairwise commutators
\begin{equation*}
\varepsilon_{jk} = \frac{1}{2} \ave{\{A_{j}, A_{k}\}} = \frac{1}{2} \tr(\{A_{j}, A_{k}\} \rho).
\end{equation*}
Ideally, since we are interested in large uncertainty, we would like our measurements to exactly anti-commute, i.e.~$\varepsilon_{jk} = 0$ for $j \neq k$.

To perform device-independent certification we need an auxiliary device (Bob), which in our case is a measurement device with $2 \cdot {M \choose 2}$ settings denoted by $B_{jk, t}$, where $j, k \in [M], j \neq k$ and $t \in \{0, 1\}$ that shares entanglement with the first device. Following the procedure proposed by Slofstra~\cite{slofstra11} we estimate the following quantity for all pairs $(j, k)$, $j \neq k$:
\begin{equation*}
\beta_{jk} :=  \ave{A_{j} \otimes (B_{jk, 0} + B_{jk, 1}) + A_{k} \otimes (B_{jk, 0} - B_{jk, 1})}.
\end{equation*}
Since this is clearly equivalent to the CHSH game, we can see the entire procedure as a combination of multiple CHSH subgames in which Alice is not told which subgame she is playing. Therefore, we can apply a standard result from~\cite{tomamichel13a}, which establishes a trade-off between the observed violation and the effective anti-commutator of the observables used by Alice (in fact, the same trade-off applies on Bob's side but since we do not want to certify the auxiliary device we do not need it). More specifically, we have
\begin{equation*}
\abs{\varepsilon_{jk}} \leq \frac{\beta_{jk}}{4} \sqrt{8 - \beta_{jk}^{2}} := c_{jk}.
\end{equation*}
While this does not allow us to find the anti-commutation matrix explicitly, we can place an upper bound on its norm. It is easy to see that $\norm{T} \leq \norm{T'}$, where
\begin{equation*}
T_{jk}' =
\begin{cases}
1 &\nbox{if} j = k,\\
c_{jk} &\nbox{otherwise.}
\end{cases}
\end{equation*}
Therefore, the observed statistics allows us to bound $\norm{T}$, which turns out to be sufficient for our applications.

For completeness, we also provide an explicit description of devices that achieve the maximum violation for all subgames. Suppose Alice and Bob share a maximally entangled state of dimension $d = 2^{\lceil \frac{M - 1}{2} \rceil}$
\begin{equation*}
\ket{\Psi}_{AB} = \frac{1}{\sqrt{d}} \sum_{k = 1}^{d} \ket{k}_{A} \ket{k}_{B}
\end{equation*}
and that their measurements are
\begin{equation*}
A_{j} = \Gamma_{j} \nbox{and} B_{jk, t} = \frac{\Gamma_{j}^{\tran} + (-1)^{t} \Gamma_{k}^{\tran}}{\sqrt{2}},
\end{equation*}
where $\{\Gamma_{j}\}$ is a set of anti-commuting observables acting on $d$-dimensional Hilbert space as defined in Section~\ref{sec:preliminaries}. It is easy to check that for every pair $(j, k), j \neq k$, we obtain
\begin{equation*}
\bramatketq{\Psi}{A_{j} \otimes (B_{jk, 0} + B_{jk, 1}) + A_{k} \otimes (B_{jk, 0} - B_{jk, 1})} = 2 \sqrt{2},
\end{equation*}
which implies $\varepsilon_{jk} = 0$. Hence, we have certified a device that performs $M$ exactly anti-commuting measurements.
\end{document}